\documentclass[aip,graphicx,jmp]{revtex4-1}
\usepackage{amsmath,amssymb,amsthm}
\usepackage[T2A]{fontenc}

\draft 

\newtheorem{corollary}{Corollary}

\newtheorem{theorem}{Theorem}

\newtheorem{lemma}{Lemma}

\begin{document}

\title{A new (1+1)-dimensional matrix k-constrained KP
hierarchy} 

\author{O.I. Chvartatskyi}
\email{alex.chvartatskyy@gmail.com}
\author{Yu.M. Sydorenko}
\email{y_sydorenko@franko.lviv.ua}
\affiliation{Ivan Franko National University of L'viv}

\date{\today}

\begin{abstract}
We introduce a new generalization of matrix (1+1)-dimensional
k-constrained KP hierarchy. The new hierarchy contains matrix
generalizations of stationary DS systems, (2+1)-dimensional modified
Korteweg-de Vries equation and the Nizhnik equation. A binary
Darboux transformation method is proposed for integration of systems
from this hierarchy.

\end{abstract}


\maketitle 
\section{Introduction}
Algebraic methods are of great importance in a soliton equations theory (see \cite{LDA,Zakharov,Zakh-Manak,solitons,March,Matveev79,Matveev}). 
 In particular, they make possible to omit analytical
difficulties that arise in the investigation of corresponding direct
and inverse scattering problems for nonlinear equations. One of the
important objects arising from algebraic approach in the soliton
theory is scalar and matrix hierarchies for nonlinear integrable
systems of Kadomtsev-Petviashvili type (KP hierarchy)
\cite{DJKM1,DJKM2,SS3,MM4,Ohta}.

The KP hierarchy and its generalizations 
play an important role in mathematical physics. 
One of such generalizations is the so-called ``KP equation with
self-consistent sources'' (KPSCS), discovered by Melnikov
\cite{M1,Mf,M4,M2,M3}. In \cite{SS,KSS,Chenga1,CY,Chenga2},
k-symmetry constraints of the KP hierarchy (k-cKP hierarchy)  which
have connections with KPSCS were investigated. 
k-cKP hierarchy contains physically relevant systems like the
nonlinear Schr\"odinger equation, the Yajima-Oikawa system, a
generalization of the Boussinesq equation, and the Melnikov system.
Multicomponent k-constraints of the KP hierarchy were introduced in
\cite{SSq} and investigated in
\cite{Oevel93,ZC,Oevel96,Aratyn97,Chau1,WLG}.

 The modified
k-constrained KP (k-cmKP) hierarchy was proposed in
\cite{CY,KSO,OC}. It contains, for example, the vector Chen-Lee-Liu
and the modified KdV (mKdV) equation. Multi-component versions of
the Kundu-Eckhaus and Gerdjikov-Ivanov equations were also obtained
in \cite{KSO}, via gauge transformations of the k-cKP, respectively
the k-cmKP hierarchy. Moreover, in \cite{MSS,6SSS},
(2+1)-dimensional extensions of the k-cKP hierarchy were introduced
and dressing methods via differential transformations were
investigated. Some systems of this hierarchy were investigated via
binary Darboux transformations in \cite{PHD,BS1}. This hierarchy was
also rediscovered recently in \cite{LZL1,LZL2}.

In this paper our aim was to generalize (1+1)-dimensional matrix
k-constrained KP hierarchy to the case of two integro-differential
operators. As a result we obtain a new bidirectional (one of the
operators in obtained hierarchy depends on two indices $k$ and $l$)
generalization of (1+1)-dimensional matrix k-constrained KP
hierarchy that we will call (1+1)-BDk-cKP hierarchy (see formulae
(\ref{ex2+1})).

This work is organized as follows. In Section 2 we
 present a short survey of results on constraints for KP
 hierarchies.
 In Section 3 we
 introduce a new (1+1)-BDk-cKP hierarchy. Members of the
 obtained hierarchy are also listed there.
  (1+1)-BDk-cKP hierarchy contains matrix generalizations of
  stationary Davey-Stewartson hierarchy, new stationary Yajima-Oikawa and Melnikov hierarchies.
  In Section 4 we consider
 dressing via binary Darboux transformation for the (1+1)-BDk-cKP hierarchy.
 As an example we consider construction of solutions for the
 matrix generalization of stationary DS system that was considered
 in Section 3.
 In the final section,
Conclusions we discuss the obtained results and mention problems for
further investigations. We also present an equation obtained from
(1+1)-BDk-cKP hierarchy that generalize vector nonlinear
Schr\"odinger system (the Manakov system).

\section{$k$-constrained KP hierarchy and its extensions}\label{kckp}

To make this paper self-contained, we briefly introduce the KP
hierarchy \cite{LDA} and its multi-component k-symmetry constraints
(k-cKP hierarchy). A Lax representation of the KP hierarchy is given
by
\begin{equation}\label{ssk}
L_{t_n}=[B_n,L],\,    \qquad n\geq1,
\end{equation}
where $L=D+U_1D^{-1}+U_2D^{-2}+\ldots$ is a scalar
pseudodifferential operator, $t_1:=x$, $D:=\frac{\partial}{\partial
x}$, and $B_n:= (L^n)_+
 := (L^n)_{\geq0}=D^n+\sum_{i=0}^{n-2}u_iD^i$ is
the differential operator part of $L^n$. The consistency condition
(zero-curvature equations), arising from the commutativity of flows
(\ref{ssk}), is
are
\begin{equation}\label{LP}
B_{n,t_k}-B_{k,t_n}+[B_n,B_k]=0.
\end{equation}

Let $B^{\tau}_n$ denote the formal transpose of $B_n$, i.e.
$B^{\tau}_n:=(-1)^nD^n+\sum_{i=0}^{n-2}(-1)^iD^iu^{\top}_i$, where
$^{\top}$ denotes the matrix transpose. We will use curly brackets
to denote the action of an operator on a function whereas, for
example, $B_n \, q$ means the composition of the operator $B_n$ and
the operator of multiplication by the function $q$. The following
formula holds for $B_nq$ and $B_n\{{q}\}$:
$B_n\{{q}\}=B_nq-(B_n{q})_{>0}.$ In the case $k=2$, $n=3$ formula
(\ref{LP}) presents a Lax pair for the Kadomtsev-Petviahvili
equation \cite{KP}. Its Lax pair was obtained in \cite{D} (see also
\cite{Zakharov}).

The multicomponent k-constraints of the KP hierarchy is given by \cite{SSq} 
\begin{equation}\label{eq1}
  L_{t_n}=[B_n,L],
  \end{equation}
with the k-symmetry reduction
\begin{equation}\label{eq2}
    L_k:=L^k=B_k+\sum_{i=1}^m\sum_{j=1}^mq_im_{ij}D^{-1}r_j=B_k+{\bf
q}{\cal M}_0D^{-1}{\bf r}^{\top},
\end{equation}
where ${\bf q}=(q_1,\ldots,q_m)$ and ${\bf r}=(r_1,\ldots,r_m)$ are
vector functions, ${\cal M}_0=(m_{ij})_{i,j=1}^m$ is a constant
$m\times m$ matrix. In the scalar case ($m=1$) we obtain
k-constrained KP hierarchy \cite{SS,KSS,Chenga1,CY,Chenga2}. The
hierarchy given by (\ref{eq1})-(\ref{eq2}) admits the Lax
representation (here $k\in{\mathbf{N}}$ is fixed):
\begin{equation}\label{hier}
  [L_k,M_n]=0,\,\,\, L_k=B_k+{\bf q}{\cal M}_0D^{-1}{\bf r}^{\top}, \quad
  M_n=\partial_{t_n}-B_n.
\end{equation}
Lax equation (\ref{hier}) is equivalent to the following system:
\begin{equation}\label{hier23}
[L_k,M_n]_{\geq0}=0,\,\, M_n\{{\bf q}\}=0,\,\,\,M_n^{\tau}\{{{\bf
r}}\}=0.
\end{equation}

Below we will also use the formal adjoint
$B^*_n:=\bar{B}^{\tau}_n=(-1)^nD^n+\sum_{i=0}^{n-2}(-1)^iD^i{u}^*_i$
of $B_n$, where $^\ast$ denotes the Hermitian conjugation (complex
conjugation and transpose).

 For $k=1$, the hierarchy given by (\ref{hier23}) is a multi-component generalization of the AKNS
hierarchy. For $k=2$ and $k=3$, one obtains vector generalizations
of the Yajima-Oikawa and Melnikov \cite{Mf,M4} hierarchies,
respectively.

In \cite{CY,KSO,OC}, a k-constrained modified KP (k-cmKP) hierarchy
was introduced and investigated. 
It contains vector generalizations of the Chen-Lee-Liu, the modified
multi-component Yajima-Oikawa and Melnikov hierarchies.

An essential extension of the k-cKP hierarchy is its
(2+1)-dimensional generalization introduced in \cite{MSS,6SSS} and rediscovered in \cite{LZL1,LZL2}.
\begin{center}
\section{A new bidirectional extension of (1+1)-dimensional k-constrained
KP ((1+1)-BDk-cKP) hierarchy}\label{extended}
\end{center}
In this section we introduce a new generalization of the
(1+1)-dimensional k-constrained KP hierarchy given by (\ref{hier})
to the case of two integro-differential operators. One of them (the
operator $L_{k,l}$ (\ref{ex2+1})  generalizes the corresponding
operator $L_k$ (\ref{hier}) and depends on two independent indices
$l$ and $k$. It leads to generalization of (1+1)-dimensional k-cKP
hierarchy (\ref{hier}) in additional direction $l$ ($l=1,2,\ldots$).
For further purposes we will use the following well-known formulae
for integral operator $h_1D^{-1}h_2$ constructed by matrix-valued
functions $h_1$ and $h_2$ and the differential operator $A$ with
matrix-valued coefficients in the algebra of pseudodifferential
operators: 
\begin{equation}\label{Sydorenko:eq25}
Ah_1{\cal D}^{-1} h_2=(Ah_1{\cal D}^{-1} h_2)_{\geq0}+ A\{h_1\}{\cal
D}^{-1} h_2,
\end{equation}
\begin{equation}\label{Sydorenko:eq26}
h_1{\cal D}^{-1} h_2 A=(h_1{\cal D}^{-1} h_2 A)_{\geq0}+ h_1{\cal
D}^{-1} [A^\tau \{h_2^{\top}\}]^\top,
\end{equation}
\begin{equation}\label{Sydorenko:eq27}
h_1{\cal D}^{-1} h_2 h_3{\cal D}^{-1} h_4=h_1D^{-1}\{h_2h_3\}{\cal
D}^{-1} h_4 -h_1{\cal D}^{-1}D^{-1}\{h_2h_3\} h_4.
\end{equation}
It is known that the Matrix KP hierarchy can be formulated by a
pseudodifferential operator:
\begin{equation}
W=I+w_1D+w_2D^2+\ldots
\end{equation}
with $N\times N$-matrix-valued coefficients $w_i$. Consider the
differential operators $\mathcal{J}_kD^k$ and
$\alpha_n\partial_{t_n}-{\tilde{\mathcal{J}}}_nD^n$,
$\alpha_n\in{\mathbb{C}}$, $n,k\in{\mathbb{N}}$, where
$\mathcal{J}_k$ and ${\tilde{\mathcal{J}}}_n$ are $N\times N$
commuting matrices (i.e.,
$[{\tilde{\mathcal{J}}}_n,\mathcal{J}_k]=0$). It is evident that the
dressed operators:
\begin{equation}
L_k:=W\mathcal{J}_kD^kW^{-1}=\mathcal{J}_kD^k+u_{k-1}D^{k-1}+\ldots
+u_0+u_{-1}D^{-1}+\ldots,\,\,
\end{equation}
and
\begin{equation}
M_n:=W(\alpha_n\partial_{t_n}-{\tilde{\mathcal{J}}}_nD^n)W^{-1}=\alpha_n\partial_{t_n}-{\tilde{\mathcal{J}}}_nD^n-v_{n-1}D^{n-1}+\ldots
+v_0+v_{-1}D^{-1}+\ldots,\,\,
\end{equation}
with $N\times N$-matrix coefficients $u_i$ and $v_j$ commute:
$[L_k,M_n]=0$. We shall impose the following reduction on operators
$L_k$ and $M_n$:
\begin{equation}
(L_k)_{<0}:=(L_{k,l})_{<0}=\!c_l\!\sum_{j=0}^l\!{\bf q}[j]\!{\cal
M}_0D^{-1}{\bf r}^{\top}[l-j],\,\,(M_{n})_{<0}=-\gamma{\bf q}{\cal
M}_0D^{-1}{\bf r}^{\top},\,\,\,\gamma,\,\ c_l\in{\mathbb{C}},
\end{equation}
where ${\bf q}$ and ${\bf r}$  are $N\times m$ matrix functions;
${\bf q}[j]$ and ${\bf r}[j]$ are matrix functions of the following
form: $ {\bf{q}}[j]:=(M_{n})^j\{{\bf q}\},\,\, {\bf
r}^{\top}[j]:=((M_n^{\tau})^j\{{\bf r}\})^{\top}.$

 As a result we obtain the following bi-directional  k-cKP
hierarchy:
\begin{equation}\label{ex2+1}
\begin{array}{l}
\!L_{k,l}\!=\!B_k\!+\!c_l\!\sum_{j=0}^l\!{\bf q}[j]\!{\cal
M}_0D^{-1}{\bf
r}^{\top}[l-j],\,\!\!B_k=\mathcal{J}_kD^k+\sum_{j=0}^{k-1}u_jD^j,\,u_j=u_j(x,t_n),\,\!l=0,\ldots
\\
M_{n}=\alpha_n\partial_{t_n}-{A}_n-\gamma{\bf q}{\cal M}_0D^{-1}{\bf
r}^{\top},\,\,\,
{A}_n={\tilde{\mathcal{J}}}_nD^n+\sum_{i=0}^{n-1}v_iD^i,\,v_i=v_i(x,t_n),
\alpha_n\in{\mathbb{C}},
\end{array}
\end{equation}
where $u_j$ and $v_i$, are $N\times N$ matrix functions.

The following theorem  holds.
\begin{theorem}\label{T1}
The Lax equation $[L_{k,l},M_{n}]=0$ is equivalent to the system:
 \begin{equation}\label{fre}
 [L_{k,l},M_{n}]_{\geq0}=0,\gamma L_{k,l}\{{\bf q}\}+c_l(M_n)^{l+1}\{{\bf q}\}=0,\,\gamma L_{k,l}^{\tau}\{{\bf{r}}\}+c_l(M_n^{\tau})^{l+1}\{{\bf{r}}\}=0.
 \end{equation}
 \end{theorem}
 \begin{proof}
 From the equality $[L_{k,l},M_{n}]=[L_{k,l},M_{n}]_{\geq0}+[L_{k,l},M_{n}]_{<0}$ we obtain that the Lax equation $[L_{k,l},M_{n}]=0$ is equivalent to the following one:
\begin{equation}
[L_{k,l},M_{n}]_{\geq0}=0,\,\,[L_{k,l},M_{n}]_{<0}=0.
\end{equation}
Thus, it is sufficient to prove that
$[L_{k,l},M_{n}]_{<0}=0\Longleftrightarrow\,\,\gamma L_{k,l}\{{\bf
q}\}+c_l(M_{n})^{l+1}\{{\bf q}\}=0,$ $ \gamma L^{\tau}_{k,l}\{{\bf
r}\}+c_l(M_{n}^{\tau})^{l+1}\{{\bf r}\}=0$. Using bi-linearity of
the commutator and explicit form (\ref{ex2+1}) of operators
$L_{k,l}$ and $M_{n}$ we obtain:
\begin{equation}\label{Frt}
\begin{array}{c}
[L_{k,l},M_{n}]_{<0}=+c_l\sum_{j=0}^l[{\bf q}[j]{\cal M}_0D^{-1}{\bf
r}^{\top}[l-j],\alpha_n\partial_{t_n}-A_n]_{<0}-\\-c_l\sum_{j=0}^l[{\bf
q}[j]{\cal M}_0D^{-1}{\bf r}^{\top}[l-j],\gamma{\bf q}{\cal
M}_0D^{-1}{\bf r}^{\top}]_{<0}-[B_k,\gamma{\bf q}{\cal
M}_0D^{-1}{\bf r}^{\top}]_{<0}.
\end{array}
\end{equation}

After direct computations of each of the three summands on the
right-hand side of formula (\ref{Frt}) we obtain:

\begin{enumerate}
\item
\begin{equation}\label{equ1}
\begin{array}{c}
c_l\sum_{j=0}^l[{\bf q}[j]{\cal M}_0D^{-1}{\bf
r}^{\top}[l-j],\alpha_n\partial_{t_n}-A_n]_{<0}=-c_l\sum_{j=0}^l\left(\alpha_n{\bf
q}_{t_n}[j]-A_n\{{\bf q}[j]\}\right)\cdot\\\cdot{\cal M}_0D^{-1}{\bf
r}^{\top}[l-j]-c_l\sum_{j=0}^l{\bf q}[j]{\cal
M}_0D^{-1}\left(\alpha_n{\bf r}^{\top}_{t_n}[l-j]+A_n^{\tau}\{{\bf
r}^{\top}[l-j]\}\right).
\end{array}
\end{equation}
Equality (\ref{equ1}) is a consequence of formulae
(\ref{Sydorenko:eq25})-(\ref{Sydorenko:eq26}).
\item

\begin{equation}\label{item2}
\begin{array}{c}
\gamma c_l\sum_{j=0}^l[{\bf q}{\cal M}_0D^{-1}{\bf r}^{\top},{\bf
q}[j]{\cal M}_0D^{-1}{\bf r}^{\top}[l-j]]_{<0}=\gamma
c_l\sum_{j=0}^l {\bf q}{\cal M}_0D^{-1}\{{\bf r}^{\top}{\bf
q}[j]\}\times\\\times{\cal M}_0D^{-1}{\bf r}^{\top}[l-j]-\gamma
c_l\sum_{j=0}^l {\bf q}{\cal M}_0D^{-1}D^{-1}\{{\bf r}^{\top}{\bf
q}[j]\}{\cal M}_0{\bf r}^{\top}[l-j]+\\- \gamma c_l\sum_{j=0}^l {\bf
q}[j]{\cal M}_0D^{-1}\{{\bf r}^{\top}[l-j] {\bf q}\}{\cal
M}_0D^{-1}{\bf r}^{\top}+\\+\gamma c_l\sum_{j=0}^l {\bf q}[j]{\cal
M}_0D^{-1}D^{-1}\{{\bf r}^{\top}[l-j] {\bf q}\}{\cal M}_0{\bf
r}^{\top}.
\end{array}
\end{equation}
Formula (\ref{item2}) follows from (\ref{Sydorenko:eq27}).
\item

\begin{equation}\label{dsa}
\begin{array}{l}
-[B_k,\gamma{\bf q}{\cal M}_0D^{-1}{\bf r}^{\top}]_{<0}=-\gamma
B_k\{{\bf q}\}{\cal M}_0D^{-1}{\bf r}^{\top}+\gamma{\bf q}{\cal
M}_0D^{-1}\left(B_k^{\tau}\{{\bf r}\})^{\top}\right).
\end{array}
\end{equation}
The latter equality is obtained via
(\ref{Sydorenko:eq25})-(\ref{Sydorenko:eq26}).
\end{enumerate}
From formulae (\ref{Frt})-(\ref{dsa}) we have
\begin{equation}
\begin{array}{l}
[L_{k,l},M_{n}]_{<0}=c_l\left(\!\sum_{j=0}^l{{\bf q}}[j]{\cal
M}_0D^{-1}M_{n}^{\tau}\{{\bf r}[l-j]\}-\sum_{j=0}^lM_n\{{\bf
q}[j]\}{\cal M}_0D^{-1}{\bf r}^{\top}[l-j]\!\right)\!-\\-\gamma
L_{k,l}\{{\bf q}\}{\cal M}_0D^{-1}{\bf r}^{\top}+\gamma{\bf q}{\cal
M}_0D^{-1} (L_{k,l}^{\tau}\{{\bf r}\})^{\top}=c_l\sum_{j=0}^l{\bf
q}[j]{\cal M}_0D^{-1}{\bf r}^{\top}[l-j+1]-\\-c_l\sum_{j=0}^l{\bf
q}[j+1]{\cal M}_0D^{-1}{\bf r}^{\top}[l-j]-\gamma L_{k,l}\{{\bf
q}\}{\cal M}_0D^{-1}{\bf r}^{\top}+\gamma{\bf q}{\cal
M}_0D^{-1}L^{\tau}_{k,l}\{{\bf r}\}^{\top}=\\=-(\gamma
L_{k,l}+c_l(M_n)^{l+1})\{{\bf q}\}{\cal M}_0D^{-1}{\bf
r}^{\top}+{\bf q}{\cal M}_0D^{-1}((\gamma
L_{k,l}^{\tau}+c_l(M_n^{\tau})^{l+1})\{{\bf r}\})^{\top}.
\end{array}
\end{equation}
From the last equality we obtain the equivalence of the equation
$[L_k,M_{n,l}]=0$ and (\ref{fre}).

\end{proof}
New hierarchy (\ref{ex2+1}) includes Matrix k-constrained
KP-hierarchy  ($\gamma=0$, $l=0$).


 The following corollary immediately follows from Theorem \ref{T1}:
\begin{corollary}
The Lax equation $[\tilde{L}_{k,l},{M}_{n}]=0$, where
\begin{equation}\label{tM}
\tilde{L}_{k,l}=\gamma L_{k,l}+c_l(M_n)^{l+1}
\end{equation}
 and the operators $\tilde{L}_{k,l}$ and $M_n$ are defined by (\ref{ex2+1}), is equivalent to the system:
 \begin{equation}\label{frex}
 [{L}_{k,l},{M}_{n}]_{\geq0}=0,\tilde{L}_{k,l}\{{\bf q}\}=0,\,\tilde{L}_{k,l}^{\tau}\{{\bf{r}}\}=0.
 \end{equation}
\end{corollary}
The BDk-cKP hierarchy (\ref{ex2+1}) admits an essential
generaliation:
\begin{equation}\label{sex2+1}
\begin{array}{l}
\!P_{k,s}\!=\!B_k\!+\sum_{l=0}^sc_l
\!\sum_{j=0}^l\!{\bf q}[j]\!{\cal M}_0D^{-1}{\bf
r}^{\top}[l-j],\\\,\!\!B_k=\mathcal{J}_kD^k+\sum_{j=0}^{k-1}u_jD^j,\,u_j=u_j(x,t_n),\,\!l=0,\ldots
\\
M_{n}=\alpha_n\partial_{t_n}-{A}_n-\gamma{\bf q}{\cal M}_0D^{-1}{\bf
r}^{\top},\,\,\,
{A}_n={\tilde{\mathcal{J}}}_nD^n+\sum_{i=0}^{n-1}v_iD^i,\,v_i=v_i(x,t_n),
\alpha_n\in{\mathbb{C}},
\end{array}
\end{equation}
\begin{corollary}
Lax equation $[P_{k,s},M_n]=0$ is equivalent to the system:
\begin{equation}
[P_{k,s},M_n]=0,\,\,(\gamma
P_{k,s}+\sum_{l=0}^sc_l(M_n)^{l+1})\{{\bf q}\}=0,\,\, (\gamma
P_{k,s}^{\tau}+\sum_{l=0}^sc_l(M_n^{\tau})^{l+1})\{{\bf r}\}=0.
\end{equation}

\end{corollary}

We do not tend to consider precisely the case $\gamma=0$ in the
hierarchy (\ref{ex2+1}) in this paper. Thus, without loss of
generality we put $\gamma=1$.

 For further convenience we will consider the Lax pairs
consisting of the operators $\tilde{L}_{k,l}$ (\ref{tM}) and
${M}_{n}$ (\ref{ex2+1}) 
(the operator $\tilde{L}_{k,l}$ is involved in equations for
functions ${\bf q}$ and ${\bf r}$; see formulae (\ref{frex})).
 Consider examples of equations given by operators $\tilde{L}_{k,l}$
 (\ref{tM})
 and ${M}_{n}$ (\ref{ex2+1})
that can be obtained under certain choice of $(k,n,l)$. For
simplicity we will also introduce notations: $t:=t_0$,
$\alpha:=\alpha_0$.





%
\begin{enumerate}
\item $k=2$, $l=1$, $n=0$.
In this case we obtain the following Lax pair in (\ref{ex2+1}):
\begin{equation}\label{1eq}
\begin{array}{l}
\tilde{{L
}}_{2,1}=L_{2,1}+c_1(M_0)^2={{D}^{2}}+v_0+c_1\alpha^2\partial^2_t-2c_1\alpha{\bf
q}{\cal M}_0D^{-1}\partial_t{\bf r}^{\top},\,\,c\in{\mathbb{C}},\,\,\,\\
{{M}_{0}}=\alpha\partial_{t}-{\bf{q}}{{\mathcal{M}}_{0}}{{D}^{-1}}{{\bf{r}}^{\top
}}.
\end{array}
\end{equation}
The commutator equation $[\tilde{L}_{2,1},{M}_{0}]=0$ is equivalent
to the system:
\begin{equation}\label{DSnr}
\begin{array}{c}
 {\bf q}_{xx}+c_1\alpha^2{\bf q}_{tt}+v_0{\bf q}+c_1{\bf
 q}{\cal M}_0S=0,\\
 {\bf r}^{\top}_{xx}+c_1\alpha^2{\bf r}^{\top}_{tt}+{\bf
 r}^{\top}v_0+c_1S{\cal M}_0{\bf r}^{\top}=0,\\
 \alpha v_{0t}=-2({\bf q}{\cal M}_0{\bf r}^{\top})_x,\,\,S_{x}=-2\alpha({\bf r}^{\top}{\bf q})_t.
\end{array}
\end{equation}
Equation (\ref{DSnr}) and its integro-differential Lax
representation in (2+1)-dimensional case was  investigated in
\cite{n}. Consider additional reductions of pair of the operators
$\tilde{L}_{2,1}$ and ${M}_{0}$ (\ref{1eq}) and system (\ref{DSnr}).
After the reduction $c_1\in\mathbb{R}$, $\alpha\in{\mathbb{R}}$,
${\bf r}^{\top}={\bf q}^*$, ${\cal M}_0={\cal M}_0^*$, the operators
$\tilde{L}_{2,1}$ and ${M}_{0}$ are Hermitian and skew-Hermitian
respectively, and (\ref{DSnr}) takes the form
\begin{equation}\label{DS}
\begin{array}{c}
 {\bf q}_{xx}+c_1\alpha^2{\bf q}_{tt}+v_0{\bf q}+c_1{\bf
 q}{\cal M}_0S=0,\\
 \alpha v_{0t}=-2({\bf q}{\cal M}_0{\bf q}^*)_x,\,\,S_{x}=-2\alpha({\bf q}^*{\bf q})_t.
\end{array}
\end{equation}



Let us consider (\ref{DS}) in the case where $\alpha=1$, $u:={\bf
q}$ and $\mu:={\cal M}_0$ are scalars. Then (\ref{DS}) can be
rewritten as
\begin{equation}\label{DSuscal}
\begin{array}{c}
 u_{xx}+c_1u_{tt}+v_0u+\mu c_1 S u=0,\,\,
 v_{0,t}=-2\mu(|u|^2)_x,\,\,S_{x}=-2(|u|^2)_t.
\end{array}
\end{equation}
As a consequence of (\ref{DSuscal}) we obtain
\begin{equation}\label{DS2}
\begin{array}{c}
 u_{xx}+c_1u_{tt}+\mu S_{1}u=0,\,\,\,
 S_{1,xt}=-2(|u|^2)_{xx}-2c_1(|u|^2)_{tt},
\end{array}
\end{equation}
where $S_1=\mu^{-1}v_0+c_1S$. This is the well-known stationary
Davey-Stewartson system (DS-I) and (\ref{DS}) is therefore a matrix
(noncommutative) generalization. The interest in noncommutative
versions of DS systems and some other noncommutative nonlinear
equations (in particular, solution generating technique) has also
arisen recently in \cite{DMH,GM,GNS,GHN}.

\item $k=2$, $l=2$, $n=0$
\begin{equation}\label{1eqka}
\begin{array}{l}
\tilde{{L }}_{2,2}=L_{2,2}+c_2(M_0)^3=c_2\alpha^3\partial_t^3+{{D}^{2}}+v_0-3\alpha^2c_2{\bf q}_t{\cal M}_0D^{-1}{\bf r}^{\top}_t  \\
      +3\alpha c_2{\bf q}{\cal M}_0\partial_tD^{-1}{\bf r}^{\top}{\bf q}{\cal M}_0D^{-1}{\bf r}^{\top}
  - 3\alpha c_2{\bf q}{\cal M}_0D^{-1}\{{\bf r}^{\top}{\bf q}\}_t{\cal M}_0D^{-1}{\bf r}^{\top} \\
  -3c_2\alpha^2\partial_t{\bf q}{\cal M}_0D^{-1}{\bf
  r}^{\top}\partial_t,\\
  {{M}_{0}}=\alpha\partial_{t}-{\bf{q}}{{\mathcal{M}}_{0}}{{D}^{-1}}{{\bf{r}}^{\top
}}.
  \end{array}
\end{equation}
In the vector case ($m=1$) after setting $\mu:={\cal
M}_0\in{\mathbb{C}}$, $\alpha=1$ the commutator equation
$[\tilde{L}_{2,2},{M}_{0}]=0$ is equivalent to the system:
\begin{equation}\label{aDSL1M3}
\begin{array}{l}
-{\bf q}_{xx}-c_2{\bf q}_{ttt}- v_0{\bf q}+ 3c_2\mu({\bf q}S_1)_t
-3c_2\mu{\bf q}S_2=0, \\
 -{\bf r}^{\top}_{xx}-c_2{\bf
r}^{\top}_{ttt}-{\bf r}^{\top}v_0+ 3c_2\mu S_1{\bf
r}^{\top}_t+3c_2\mu S_2{\bf r}^{\top}=0, \\
v_{0t}=-2\mu({\bf q}{\bf r}^{\top})_x,\,\,
S_{1x}=({\bf{r}}^{\top}{\bf q})_t,\,\,\,\,S_{2x}=({\bf
r}^{\top}_t{\bf q})_t.
\end{array}
\end{equation}
Thus, equation given by commutator $[\tilde{L}_{2,2},M_0]=0$
generalizes (\ref{aDSL1M3}) to the matrix case but we do not present
it because of its rather complicated structure.




\item $k=3,l=1$, $n=0$.

In this case the operator $\tilde{L}_{3,1}$ has the form:
\begin{equation}\label{L1M3e}
\begin{array}{l}
\tilde{L}_{3,1}=L_{3,1}-c_1(M_0)^2= {D}^3+ v_{1} {D}+
v_{0}+c_1\alpha^2\partial^2_t-2c_1\alpha{\bf q}{\cal M}_0D^{-1}{\bf
r}^{\top}_t-\\-2c_1\alpha{\bf q}{\cal
M}_0D^{-1}{\bf r}^{\top}\partial_t,\\
{{M}_{0}}=\alpha\partial_{t}-{\bf{q}}{{\mathcal{M}}_{0}}{{D}^{-1}}{{\bf{r}}^{\top
}}.
\end{array}
\end{equation}\

In the vector case ($N=1$) the equation $[\tilde{L}_{3,1},M_0]=0$ is
equivalent to the system:
\begin{equation}\label{mKDV2}
\begin{array}{l}
{\bf q}_{xxx}+c_1\alpha^2{\bf q}_{tt}+ v_{1}{\bf q}_{x} + v_{0}{\bf
q}+c_1{\bf
 q}{\cal M}_0S_1=0,\\
-  {\bf r}^{\top}_{xxx}+c_1\alpha^2{\bf r}^{\top}_{tt} -({\bf
r}^{\top}v_{1})_x + {\bf r}^{\top}v_{0}+c_1S_1{\cal M}_0{\bf r}^{\top}=0,\\
\alpha{v_{0,t}}
 = - 3({\bf q}_{x}{\cal M}_0{\bf r}^{\top})_{x},\,\,\alpha  {v_{1,t}} =-3 ({\bf q}{\cal M}_0{\bf
r}^{\top})_{x},\,\,S_{1x}=-2\alpha({\bf r}^{\top}{\bf q})_t.
\end{array}
\end{equation}
\item $k=3,l=2$,$n=0$.

\begin{equation}\label{L1M3}
\begin{array}{l}
 \tilde{L}_{3,2}=D^3+c_2\alpha^3\partial_t^3-v_1D+
v_0-3\alpha^2c_2{\bf q}_t{\cal M}_0D^{-1}{\bf r}^{\top}_t  \\
      +3\alpha c_2{\bf q}{\cal M}_0\partial_tD^{-1}{\bf r}^{\top}{\bf q}{\cal M}_0D^{-1}{\bf r}^{\top}
  - 3\alpha c_2{\bf q}{\cal M}_0D^{-1}\{{\bf r}^{\top}{\bf q}\}_t{\cal M}_0D^{-1}{\bf r}^{\top} \\
  -3c_2\alpha^2\partial_t{\bf q}{\cal M}_0D^{-1}{\bf
  r}^{\top}\partial_t,\\
  M_0=\alpha\partial_{t}-{\bf{q}}{{\mathcal{M}}_{0}}{{D}^{-1}}{{\bf{r}}^{\top
}}.
\end{array}
\end{equation}
The Lax equation $[\tilde{L}_{3,2},M_0]=0$ results in the
(1+1)-dimensional matrix mKdV-type system that has rather
complicated form. For this reason we will consider only some special
cases of it (matrix generalization of (2+1)-dimensional mKdV system
and its Lax representation can be found in \cite{n}):

\begin{enumerate}
\item Consider the scalar case of the pair (\ref{L1M3}) (i.e., $N=m=1$), setting
${\mathbb{R}}\ni\mu:={\cal M}_0$, $q(x,t):={\bf q}(x,t)$, ${
r}(x,t):={\bf r}(x,t)$ and $\alpha=1$. Under additional Hermitian
conjugation reduction: $c_2\in{\mathbb{R}}$, ${r}={\bar{q}}$,
Lax equation $[\tilde{L}_{3,2},M_0]=0$ is equivalent to the
equation:
\begin{eqnarray}\label{DSL1M3rd}
&& {q}_{xxx}+c_2{q}_{ttt}-
3\mu{q}_x\int|q|^2_xdt- 3c_2\mu{q}_t\int|q|^2_tdx-\nonumber \\
&&-3c_2\mu{q}\int({\bar{q}}q_t)_tdx-3\mu{q}\int({q}_xq)_xdt=0.
\end{eqnarray}
after setting $t=x$, $q=\bar{q}$ and $c_2=-2$ (\ref{DSL1M3rd}) takes
the form
\begin{equation}\label{DSL1M3rdddd}
{q}_{xxx}- 6\mu q^2q_x=0,
\end{equation}
which is the stationary mKdV equation. The system 
(\ref{DSL1M3rd}) is its complex spatially two-dimensional
generalization.

\item Consider the scalar case ($N=1$, $m=1$) of the Lax pair given by (\ref{L1M3}) under additional reduction $\beta=1$, $\mu:={\cal M}_0=1$, $r:={\bf r}=\nu$ with a constant $\nu{\in{\mathbb{R}}}$.
In terms of $u := {\bf q}\nu$ Lax equation $[\tilde{L}_{3,2},M_0]=0$
is equivalent to the following one:
\begin{equation}\label{DSL1M3r2scal}
u_{xxx}+c_2u_{ttt}-3D\left\{\left(\int u_x
dt\right)u\right\}-3c_2\partial_t\left\{u\left(\int u_t
dx\right)\right\}=0,
\end{equation}
which is the stationary case of Nizhnik equation \cite{Nizhnik80}.
\end{enumerate}
\end{enumerate}

\begin{center}
\section{Dressing methods for the new bidirectional (1+1)-dimensional k-constrained KP
hierarchy}\label{dressed}
\end{center}

In this section our aim is to consider hierarchy of equations given
by the Lax pair (\ref{ex2+1}) in case $\gamma=1$.
We suppose that the operators $L_{k,l}$ and $M_{n}$ in (\ref{ex2+1})
satisfy the commutator equation $[L_{k,l},M_{n}]=0$. At first we
recall some results from \cite{K2009}. Let $N\times K$-matrix
functions $\varphi$ and $\psi$ be solutions of linear problems:
\begin{equation}\label{pr}
\begin{array}{c}
M_{n}\{\varphi\}=\varphi\Lambda,\,\,M_n^{\tau}\{\psi\}=\psi\tilde{\Lambda},\,\,\Lambda,\tilde{\Lambda}\in Mat_{K\times K}({\mathbb{C}}).\\
\end{array}
\end{equation}
Introduce  binary Darboux transformation (BDT) in the following way:
\begin{equation}\label{W}
W=I-\varphi\left(C+D^{-1}\{\psi^{\top}\varphi\}\right)^{-1}D^{-1}\psi^{\top},
\end{equation}
where $C$ is a $K\times K$-constant nondegenerate matrix. The
inverse operator $W^{-1}$ has the form:
\begin{equation}\label{W-}
W^{-1}=I+\varphi
D^{-1}\left(C+D^{-1}\{\psi^{\top}\varphi\}\right)^{-1}\psi^{\top}.
\end{equation}
 The following theorem is proven in \cite{K2009}.
\begin{theorem}{\cite{K2009}}\label{2009}
The operator $\hat{M}_n:=WM_nW^{-1}$  obtained from $M_n$ in
(\ref{ex2+1}) via BDT (\ref{W}) has the form
\begin{equation}\label{Lop}
\hat{M}_n:=WM_nW^{-1}=\alpha_n\partial_{t_n}-\hat{A}_n-\hat{\bf
q}{\cal M}_0D^{-1}{\hat{{\bf r}}}^{\top}+\Phi{\cal
M}_1D^{-1}\Psi^{\top},\,
\hat{A}_n={\tilde{\mathcal{J}}}_nD^n+\sum_{j=0}^{n-1}\hat{v}_jD^j,
\end{equation}
where
\begin{equation}\label{DSM}
\begin{array}{l}
{\cal M}_1=C\Lambda-\tilde{\Lambda}^{\top}C,\,
\Phi=\varphi\Delta^{-1},\,\,
\Psi=\psi\Delta^{-1,\top},\,\Delta=C+D^{-1}\{\psi^{\top}\varphi\},\\
{\hat{\bf q}}=W\{{\bf q}\},\,\,{\hat{\bf r}}=W^{-1,\tau}\{{\bf r}\}.
\end{array}
\end{equation}
 $\hat{v}_j$ are $N\times N$-matrix coefficients depending on functions $\varphi$,
$\psi$ and $v_j$.
\end{theorem}
Exact forms of all coefficients  $\hat{v}_j$ are given in
\cite{K2009}.

The following corollary follows from Theorem \ref{2009}:
\begin{corollary}
The functions $\Phi=\varphi\Delta^{-1}=W\{\varphi\}C^{-1}$ and $\Psi
=\psi\Delta^{-1,\top}=W^{-1,\tau}\{\psi\}C^{\top,-1}$ satisfy the
equations
\begin{equation}\label{ro}
\hat{M}_n\{\Phi\}=\Phi C\Lambda C^{-1},\,\,
\hat{M}_n^{\tau}\{\Psi\}=\Psi C^{\top}\tilde{\Lambda}C^{\top,-1}.
\end{equation}
\end{corollary}
 For further purposes we will need the following lemmas.
\begin{lemma}
Let ${\cal M}_{l+1}$ be a matrix of the form
\begin{equation}
{\cal M}_{l+1}=C\Lambda^{l+1}-(\tilde{\Lambda}^{\top})^{l+1}C,\,\,
l\in{\mathbb{N}}.
\end{equation}
The following formula holds:
\begin{equation}\label{lemma11}
{\cal M}_{l+1}=\sum_{s=0}^lC\Lambda^sC^{-1}{\cal
M}_1C^{-1}(\tilde{\Lambda}^{\top})^{l-s}C.
\end{equation}
\end{lemma}
\begin{proof}
The following recurrent formulae that can easily be checked by
direct calculation:
\begin{equation}\label{mw}
{\cal M}_{2}=C\Lambda C^{-1}{\cal M}_1+{\cal M}_1
C\tilde{\Lambda}^{\top}C^{-1},
\end{equation}
\begin{equation}\label{vc}
{\cal M}_{l+1}=C\Lambda C^{-1}{\cal M}_{l}+{\cal
M}_{l}C^{-1}{\tilde{\Lambda}}^{\top}C-C\Lambda C^{-1}{\cal
M}_{l-1}C^{-1}{\tilde{\Lambda}}^{\top}C.
\end{equation}
Using formulae (\ref{mw})-(\ref{vc}) and induction by $k$ we can
prove that the following formula holds:
\begin{equation}\label{fla}
{\cal M}_{l+1}\!\!=\!\!\sum_{s=0}^kC\Lambda^sC^{-1}{\cal
M}_{l-k+1}C^{-1}({\tilde{\Lambda}}^{\top})^{k-s}\!C\!-\!\sum_{s=1}^kC\Lambda^sC^{-1}{\cal
M}_{l-k}C^{-1}({\tilde{\Lambda}}^{\top})^{k-s+1}C,k\leq l\!-\!2,
\end{equation}
for some $k\leq l-2$.
After the substitution of $k=l-2$ in (\ref{fla}) and using
(\ref{mw}) we can obtain formula (\ref{lemma11}).
This finishes the proof of formula (\ref{lemma11}) and Lemma 1.

\end{proof}

\begin{lemma}
The following formula
\begin{equation}
\Phi{\cal M}_{l+1}D^{-1}\Psi^{\top}=\sum_{s=0}^{l}\Phi[s]{\cal
M}_{1}D^{-1}\Psi^{\top}[l-s],
\end{equation}
holds, where
\begin{equation}\label{fk}
\Phi[j]:=(\hat{M}_n)^j\{\Phi\},\,\Psi[j]:=(\hat{M}_n^{\tau})^j\{\Psi\}.
\end{equation}
\end{lemma}
\begin{proof}
Lemma 2 is a consequence of Corollary 1 and formula (\ref{lemma11})
of Lemma 1. Namely, the following relations hold:
\begin{equation}\nonumber
\Phi{\cal
M}_{l+1}D^{-1}\Psi^{\top}=\sum_{s=0}^l{\Phi}C\Lambda^sC^{-1}{\cal
M}_1C^{-1}D^{-1}(\tilde{\Lambda}^{\top})^{l-s}C\Psi^{\top}=\sum_{s=0}^{l}\Phi[s]{\cal
M}_{1}D^{-1}\Psi^{\top}[l-s].
\end{equation}
\end{proof}
Now we assume that the functions $\varphi$ and $\psi$ in addition to
equations (\ref{pr}) satisfy the equations:
\begin{equation}\label{prm}
L_{k,l}\{\varphi\}=-c_l\varphi\Lambda^{l+1}=-c_lM_n^{l+1}\{\varphi\},\,\,L_{k,l}^{\tau}\{\psi\}=-c_l\psi\tilde{\Lambda}^{l+1}=-c_l(M_n^{\tau})^{l+1}\{\psi\}.
\end{equation}
Problems (\ref{prm}) can be rewritten via the operator
$\tilde{L}_{k,l}$ (\ref{tM}) as:
\begin{equation}
\tilde{L}_{k,l}\{\varphi\}=0,\,\,\,
\tilde{L}_{k,l}^{\tau}\{\psi\}=0.
\end{equation}
 The following theorem for the operators $L_{k,l}$ (\ref{ex2+1}) and
 $\tilde{L}_{k,l}$ (\ref{tM}) holds:
\begin{theorem}\label{M}
Let $N\times K$ -matrix functions $\varphi$, $\psi$ be solutions of
problems (\ref{pr}) and (\ref{prm}). The transformed operator
$\hat{L}_{k,l}:=WL_{k,l}W^{-1}$ obtained via BDT $W$ (\ref{W}) has
the form:
\begin{equation}\label{Mop}
\begin{array}{l}
\hat{L}_{k,l}:=WL_{k,l}W^{-1}=\hat{B}_k+c_l\sum_{j=0}^l\hat{{\bf
q}}[j]{\cal M}_0D^{-1}\hat{{\bf
r}}^{\top}[l-j]+\\+c_l\sum_{s=0}^{l}\Phi[s]{\cal
M}_{1}D^{-1}\Psi^{\top}[l-s],\,\,\hat{B}_k=\mathcal{J}_kD^k+\sum_{i=0}^{k-1}\hat{u}_iD^i,
\end{array}
\end{equation}
where the matrix ${\cal M}_n$ and the functions $\hat{{\bf q}}$,
$\hat{{\bf r}}$, $\Phi[s]$, $\Psi[l-s]$ are defined by formulae
(\ref{DSM}), (\ref{fk}) and $\hat{{\bf q}}[j]$, $\hat{{\bf r}[j]}$
have the form
\begin{equation}
\hat{{\bf q}}[j]=(\hat{M}_n^j)\{\hat{\bf q}\},\,\,\, \hat{{\bf
r}}[j]=(\hat{M}_n^{j})^{\tau}\{{\hat{\bf r}}\},
\end{equation}
 $\hat{v}_i$ are
$N\times N$-matrix coefficients that depend on the functions
$\varphi$, $\psi$ and $v_i$. The transformed operator
${\hat{\tilde{L}}}_{k,l}=W\tilde{L}_{k,l}W^{-1}$ has the form:
\begin{equation}\label{tmh}
\hat{{\tilde{L}}}_{k,l}=W\tilde{L}_{k,l}W^{-1}=\hat{L}_{k,l}+c_l(\hat{M}_n)^{l+1},
\end{equation}
where $\hat{M}_n$ is given by (\ref{Lop}).
\end{theorem}
\begin{proof}
We shall rewrite the operator $L_{k,l}$ (\ref{ex2+1}) in the form
\begin{equation}\label{Mn1ns}
L_{k,l}=\mathcal{J}_kD^k+\sum_{i=0}^{k-1}u_iD^i+c_l\tilde{{\bf
q}}\tilde{{\cal M}}_0D^{-1}\tilde{\bf r}^{\top},
\end{equation}
where $\tilde{{\cal M}}_0$ is an $m(l+1)\times m(l+1)$-
block-diagonal matrix with entries of ${\cal M}_0$ at the diagonal;
$\tilde{{\bf q}}:=({\bf q}[0],{\bf q}[1],\ldots,{\bf q}[l])$,
$\tilde{{\bf r}}:=({\bf r}[l],{\bf r}[l-1],\ldots,{\bf r}[0])$.
Using Theorem \ref{2009} we obtain that
\begin{equation}\label{fgre}
\hat{L}_{k,l}=\mathcal{J}_kD^k+\sum_{i=0}^{k-1}\hat{u}_iD^i+c_l\hat{\tilde{{\bf
q}}}\tilde{{\cal M}}_0D^{-1}\hat{\tilde{{\bf r}}}^{\top} +\Phi{\cal
M}_{l+1}D^{-1}\Psi^{\top},
 \end{equation}
where $\hat{\tilde{{\bf q}}}=W\{\tilde{{\bf q}}\}$,
$\hat{\tilde{{\bf q}}}=W^{-1,\tau}\{{\tilde{\bf r}}\}$. Using the
exact form of ${\tilde {\bf q}}$ and $\tilde{\bf r}$ we have
\begin{equation}
\hat{\tilde{{\bf q}}}=W\{\tilde{{{\bf q}}}\}=(W\{{\bf
q}[0]\},\ldots, W\{{\bf q}[l]\}),\, \hat{\tilde{{\bf
r}}}=W^{-1,\tau}\{\tilde{{{\bf r}}}\}=(W^{-1,\tau}\{{\bf
r}[l]\},\ldots, W^{-1,\tau}\{{\bf r}[0]\}).
\end{equation}
We observe that
\begin{equation}
W\{{\bf q}[i]\}=WL^i\{{\bf q}\}=WL^iW^{-1}\{W\{{\bf
q}\}\}=\hat{L}^i\{\hat{{\bf q}}\}=:\hat{{\bf q}}[i].
\end{equation}
It can be shown analogously that
 $W^{-1,\tau}\{{\bf
r}[i]\}=\hat{L}^{\tau,i}\{W^{-1,\tau}\{{\bf
r}\}\}=\hat{L}^{\tau,i}\{\hat{\bf r}\}=:\hat{\bf r}[i]$. Thus we
have:
\begin{equation}\label{sd}
 \hat{\tilde{{\bf
q}}}\tilde{{\cal M}}_0D^{-1}\hat{\tilde{{\bf
r}}}^{\top}=\sum_{j=0}^l\hat{{\bf q}}[j]{\cal M}_0D^{-1}\hat{{\bf
r}}^{\top}[l-j].
\end{equation}
For the last item in (\ref{fgre}) from Lemma 2 we have:
\begin{equation}\label{sdd}
\Phi{\cal M}_{l+1}D^{-1}\Psi^{\top}=\sum_{s=0}^{l}\Phi[s]{\cal
M}_{1}D^{-1}\Psi^{\top}[l-s].
\end{equation}
Using formulae (\ref{fgre}), (\ref{sd}), (\ref{sdd}) we obtain that
the operator $\hat{M}_{n,l}$ has form (\ref{Mop}). The exact form of
the operator $\hat{\tilde{M}}_{n,l}$ follows from formula
(\ref{Mop}) and Theorem \ref{2009}.
\end{proof}
From Theorem \ref{M} we obtain the following corollary.
\begin{corollary}\label{Corol}
Assume that functions $\varphi$ and $\psi$ satisfy problems
(\ref{pr}) and (\ref{prm}). Then the functions
$\Phi=W\{\varphi\}C^{-1}$ and $\Psi=W^{-1,\tau}\{\psi\}C^{\top,-1}$
(see formulae (\ref{DSM})) satisfy the equations:
\begin{equation}\label{meq}
\hat{\tilde{L}}_{k,l}\{\Phi\}=\hat{L}_{k,l}\{\Phi\}+c_l(\hat{M}_n)^{l+1}\{\Phi\}=0,
\,\,\hat{\tilde{L}}_{k,l}^{\tau}\{\Psi\}=\hat{L}^{\tau}_{k,l}\{\Psi\}+c_l(\hat{M}^{\tau}_n)^{l+1}\{\Psi\}=0,
\end{equation}
where the operators $\hat{L}_{k,l}$, $\hat{\tilde{L}}_{k,l}$ and
$\hat{M}_{n}$
 are defined by (\ref{Lop}), (\ref{Mop}) and
(\ref{tmh}).
\end{corollary}

As an example we will consider dressing methods for equations
connected with the operators $\tilde{L}_{2,1}$, $M_0$. Assume that
$\varphi$ and $\psi$ are $N\times K$-matrix functions that satisfy
the equations
\begin{equation}\label{Leq1}
 M_0\{\varphi\}=\varphi\Lambda,\,\,M_0^{\tau}\{\psi\}=\psi\tilde{\Lambda},\,\,M_0:=\alpha\partial_t.\\
 \end{equation}
 By Theorem \ref{2009} we obtain that the dressed operator $\hat{M}_0$
 via BDT $W$ (\ref{W}) has the form
 \begin{equation}\label{hL0}
 \hat{M}_0=WM_0W^{-1}=\alpha\partial_t+\Phi{\cal
 M}_1D^{-1}\Psi^{\top}.
 \end{equation}
Assume that $N\times K$-matrix functions $\varphi$ and $\psi$ in
addition to equations (\ref{Leq1}) also satisfy the equations
\begin{equation}\label{M2e}
 {L}_{2,1}\{\varphi\}=-c_1\varphi\Lambda^2=-c_1(M_0)^2\{\varphi\},\,\,{L}_{2,1}^{\tau}\{\psi\}=-c_1\psi{\tilde\Lambda}^2=-c_1(M_0^{\tau})^2\{\psi\},\,
 {{L}}_{2,1}:=D^2.
\end{equation}
By Theorem \ref{M} we obtain that the transformed operator
$\hat{L}_{2,1}$ has the form
\begin{equation}
\hat{{{L}}}_{2,1}=W{L}_{2,1}W^{-1}=D^2+\hat{v}_0+\hat{M}_0\{\Phi\}{\cal
 M}_1D^{-1}\Psi^{\top}+\Phi{\cal
 M}_1D^{-1}((\hat{M}_0^{\tau})\{\Psi\})^{\top}.
\end{equation}
By direct calculations it can be obtained that
$\hat{v}_0=2(\varphi\Delta^{-1}\psi^{\top})_x$,
$\Delta=C+D^{-1}\{\psi^{\top}\varphi\}$. It can be easily checked
that
\begin{equation}\label{a}
\begin{array}{l}
\alpha(\varphi\Delta^{-1}\psi^{\top})_t=\alpha\varphi_t\Delta^{-1}\psi^{\top}-\alpha\varphi\Delta^{-1}D^{-1}\{\psi^{\top}\varphi\}_t\Delta^{-1}\psi^{\top}+\alpha\varphi\Delta^{-1}\psi^{\top}_t=\\=
\varphi\Delta^{-1}(C\Lambda+\alpha
D^{-1}\{\psi^{\top}\varphi_t\})\Delta^{-1}\psi^{\top}-\alpha\varphi\Delta^{-1}D^{-1}\{\psi^{\top}\varphi\}_t\Delta^{-1}\psi^{\top}+\\+\varphi\Delta^{-1}
(-\tilde{\Lambda}^{\top}C+\alpha
D^{-1}\{\psi^{\top}_t\varphi\})\Delta^{-1}\psi^{\top}=\Phi {\cal
M}_1\Psi^{\top}.
\end{array}
\end{equation}
From the latter formula we obtain that
\begin{equation}\label{adS}
\alpha\hat{v}_{0t}=2\alpha(\varphi\Delta^{-1}\psi^{\top})_{xt}=2(\Phi{\cal
M}_1\Psi^{\top})_x.
\end{equation}
From Corollary \ref{Corol} we see that the functions
$\Phi=\varphi\Delta^{-1}$ and $\Psi$$=\psi\Delta^{\top,-1}$ where
$\Delta=C+D^{-1}\{\psi^{\top}\varphi\}$ (see formulae (\ref{DSM}))
 satisfy equations (\ref{meq}). After the change ${\bf q}:=\Phi$,
${\bf r}:=\Psi$, ${\cal M}_0:=-{\cal M}_1$, $v_0:=\hat{v}_0$ from
formulae (\ref{meq}) and (\ref{adS}) we obtain that $N\times
K$-matrix functions ${\bf q}$, ${\bf r}$, an $N\times N$-matrix
function $v_0$, a $K\times K$-matrix function $S=2\alpha
(\Delta^{-1})_t$ and a $K\times K$-matrix ${\cal M}_0$ satisfy
equations (\ref{DSnr}). It can be checked that in the case of
additional reductions in formulae (\ref{Leq1})-(\ref{M2e}):
$\alpha\in{\mathbb{R}}$, $c_1\in{\mathbb{R}}$,
$\tilde{\Lambda}=-\bar{\Lambda}$, $\psi=\bar{\varphi}$ and $C=C^*$
in gauge transformation operator $W$ (\ref{W}) it can be checked by
direct calculations that the functions ${\bf q}:=\Phi$,
$S_1=2\alpha(\Delta^{-1})_t$ and
$v_0=\hat{v}_0=2(\varphi\Delta^{-1}\varphi^*)_x$ satisfy matrix DS
system (\ref{DSnr})  with ${\cal M}_0=-{\cal M}_1$.

From the previous considerations we obtain that in the scalar case
($N=1$, $m=1$), $\mu:={\cal M}_0=-{\cal
M}_1=-C(\Lambda+\bar{\Lambda})$ under condition $\alpha=1$ functions
\begin{equation}\label{1sol}
u=q=\frac{\exp({\theta})}{\Delta},\,\,
S=-2\frac{\rm{Re}({\Lambda})\exp(2\rm{Re}(\theta))}{\rm{Re}(i\sqrt{c_1}\Lambda)\Delta^2},\,\,
v_0=-2\mu\frac{\rm{Re}(i\sqrt{c_1}\Lambda)\exp(2\rm{Re}(\theta))}{\rm{Re}({\Lambda})\
\Delta^2},\,
\end{equation}
where
$\Delta=-\frac{\mu}{2\rm{Re}({\Lambda})}+\frac{1}{2\rm{Re}(i\sqrt{c_1}\Lambda)}
\exp{(2\rm{Re}(\theta))}$ and $\theta=i\sqrt{c_1}\Lambda
x+
{\Lambda}t$, satisfy the scalar DS system (\ref{DSuscal}) (see also
(\ref{DS2})):
\begin{equation}\label{1solx}
\begin{array}{c}
 {u}_{xx}+
 c_1{u}_{tt}+v_0{u}+\mu c_1S{u}=0,\\
v_{0t}=-2\mu|{u}|^2_x,\,\,S_{x}=-2|u|^2_t.
\end{array}
\end{equation}
Functions $u$ and $S_1=\mu^{-1}v_0+c_1S$ are therefore solutions of
differential consequence (\ref{1solx}):

\begin{equation}\label{DS22}
\begin{array}{c}
 u_{xx}+c_1u_{tt}+\mu S_{1}u=0,\,\,\,
 S_{1,xt}=-2(|u|^2)_{xx}-2c_1(|u|^2)_{tt},
\end{array}
\end{equation}
Consider special cases of (\ref{DS22}) and its solutions:
\begin{enumerate}
\item $c_1=1$.
\begin{enumerate}
\item $\mu=1$. In this case functions $u$ and
$S_1=\mu^{-1}v_0+c_1S$, where $v_0$ and $S_1$ are defined by
(\ref{1sol}) represent regular solutions of (\ref{DS22}) in case
$(\rm{Re}(\Lambda))(\rm{Im}(\Lambda))>0$ (in case
$(\rm{Re}(\Lambda))(\rm{Im}(\Lambda))<0$ $u$ and $S_1$ are singular)
\item $\mu=-1$. In this case functions $u$ and
$S_1=\mu^{-1}v_0+c_1S$, where $v_0$ and $S_1$ are defined by
(\ref{1sol}) represent regular solutions of (\ref{DS22}) in case
$(\rm{Re}(\Lambda))(\rm{Im}(\Lambda))<0$ (in case
$(\rm{Re}(\Lambda))(\rm{Im}(\Lambda))>0$ $u$ and $S_1$ are singular)
\end{enumerate}
\item $c_1=-1$.
\begin{enumerate}
\item $\mu=1$. In this case functions $u$ and
$S_1=\mu^{-1}v_0+c_1S$, where $v_0$ and $S_1$ are given by
(\ref{1sol}) represent regular solutions of (\ref{DS22}).
\item $\mu=-1$. In this case functions $u$ and
$S_1=\mu^{-1}v_0+c_1S$, where $v_0$ and $S_1$ are defined by
(\ref{1sol}) represent singular solutions of (\ref{DS22}).
\end{enumerate}
\end{enumerate}

 The construction of wider
classes of solutions (e.g., soliton solutions) for vector and matrix
nonlinear systems from (1+1)-BDk-cKP hierarchy will take too much
space in this paper. Corresponding ideas can be found in
\cite{BS1,PHD,SCD}.


\section{Conclusions}
In this paper we introduced a new (1+1)-BDk-cKP hierarchy
(\ref{ex2+1}) that generalizes matrix k-constrained KP hierarchy
given by (\ref{eq1}) and (\ref{eq2}) that was investigated in
\cite{SS,KSS,Chenga1,CY,Chenga2}. We shall point that an important
case of hierarchy (\ref{ex2+1}) ($\gamma=0$) is not precisely
investigated in this paper. In particular, dressing methods for this
case still have to be elaborated. As an example let us consider the
case $\gamma=0$, $k=1$, $s=1$, $n=2$ of hierarchy (\ref{ex2+1}).
 Then
operators ${P}_{k,s}$ and $M_n$ (\ref{sex2+1}) take the form:
\begin{equation}\nonumber
\begin{array}{l}
{P}_{1,1}=D+c_1\left(\alpha_2{\bf q}_{t_2}{\cal M}_0D^{-1}{\bf
r}^{\top}-\alpha_2{\bf q}{\cal M}_0D^{-1}{\bf r}_{t_2}^{\top}-{\bf
q}_{xx}{\cal M}_0D^{-1}{\bf r}^{\top}\right.-\\\left.-{\bf q}{\cal
M}_0D^{-1}{\bf r}^{\top}_{xx}-u{\bf q}{\cal M}_0D^{-1}{\bf
r}^{\top}-{\bf q}{\cal M}_0D^{-1}{\bf r}^{\top}u\right)+c_0{\bf
q}{\cal M}_0D^{-1}{\bf r}^{\top},\\M_2=\alpha_2\partial_{t_2}-D^2-u.
\end{array}
\end{equation}
According to Corollary 2 operator equation $[{P}_{1,1},M_2]=0$ is
equivalent to the system:
\begin{equation}\label{PM}
[P_{1,1},M_2]_{\geq0}=0,\, c_1M_2^2\{{\bf q}\}+c_0M_2\{{\bf
q}\}=0,\,\,c_1(M_2^{\tau})^2\{{\bf r}\}+c_0M_2^{\tau}\{{\bf r}\}=0.
\end{equation}
System (\ref{PM}) in the vector case ($N=1$) the can be rewritten in
the following form:
\begin{eqnarray}\nonumber
&&c_1(\alpha_2^2{\bf q}_{t_2t_2}-2\alpha_2{\bf q}_{xxt_2}+{\bf
q}_{xxxx}-\alpha_2(u{\bf q})_{t_2}+(u{\bf q})_{xx}-\alpha_2u{\bf
q}_{t_2}+u{\bf q}_{xx}+u^2{\bf q})+
\\&&+c_0(\alpha_2{\bf q}_{t_2}-{\bf q}_{xx}-u{\bf q})=0,\nonumber\\
&&c_1(\alpha_2^2{\bf r}_{t_2t_2}+2\alpha_2{\bf r}_{xxt_2}+{\bf
r}_{xxxx}+\alpha_2(u{\bf r})_{t_2}+(u{\bf r})_{xx}+\alpha_2u{\bf
r}_{t_2}+u{\bf
r}_{xx}+u^2{\bf r})+\\&&c_0(-\alpha_2{\bf r}_{t_2}-{\bf r}_{xx}-u{\bf r})=0,\label{TS}\\
&&u=2\left(\frac{c_1(\alpha_2{\bf q}_{t_2}{\cal M}_0{\bf
r}^{\top}-\alpha_2{\bf q}{\cal M}_0{\bf r}^{\top}_{t_2}-{\bf
q}_{xx}{\cal M}_0{\bf r}^{\top}-{\bf q}{\cal M}_0{\bf
r}^{\top}_{xx})+c_0{\bf q}{\cal M}_0{\bf r}^{\top}}{1+4c_1{\bf
q}{\cal M}_0{\bf r}^{\top}}\right)\nonumber
\end{eqnarray}
This system is a generalization of the vector nonlinear
Schr\"{o}dinger system (NLS) in $l$-direction ($l=1$). NLS can be
obtained from the latter system in particular case ($c_1=0$) after
additional Hermitian conjugation reduction ${\bf q}=\bar{\bf r}$,
$\alpha_2\in i{\mathbb{R}}$. Analogous generalizations in
$l$-direction can be made for Yajima-Oikawa hierarchy ($k=2$, $s=0$,
$\gamma=0$ in the hierarchy (\ref{sex2+1})) and Melnikov hierarchy
($k=3$, $s=0$, $\gamma=0$ in the hierarchy (\ref{sex2+1})).
Investigations in this direction will be made in another paper.

\end{document}